\documentclass[12pt]{article}

\usepackage{amsmath,amssymb,amsfonts,mathrsfs,mathtools,dsfont,mathrsfs, amsthm, bm, bbm}
\DeclareMathAlphabet{\mathpzc}{OT1}{pzc}{m}{it}
\DeclareSymbolFontAlphabet{\amsmathbb}{AMSb}%

\usepackage{bm,bbm,dsfont,epsfig,rotating,setspace,latexsym,amsmath,epsf,amsthm,amssymb,amsfonts,epstopdf,graphicx}

\usepackage{cite,authblk}
\usepackage[dvipsnames]{xcolor}

\usepackage[normalem]{ulem}
\usepackage{comment}

\usepackage[dvipsnames]{xcolor}
\usepackage[font=small]{caption}
\usepackage[caption=false, font=small]{subfig}
\usepackage[hyphens]{url}
\usepackage[breaklinks, colorlinks, linkcolor=MidnightBlue, anchorcolor=MidnightBlue, citecolor=MidnightBlue, urlcolor=MidnightBlue]{hyperref}
\usepackage{subcaption}

\usepackage{array}

\allowdisplaybreaks

\usepackage[dvipsnames]{xcolor}
\usepackage{enumerate,enumitem}
\usepackage{csvsimple,booktabs}
\usepackage{float,subfig}
\usepackage{graphicx,epsfig,xcolor}
\usepackage{multirow}

\usepackage{fullpage}

\usepackage{ebgaramond}
\usepackage[OT1]{fontenc}
\usepackage[utf8]{inputenc}

\usepackage{thmtools}
\declaretheorem{theorem}

\declaretheorem{assumption}

\declaretheorem[sibling=theorem]{lemma}

\newcommand{\beq}{\begin{equation}}
\newcommand{\eeq}{\end{equation}}
\newcommand{\beqa}{\begin{eqnarray}}
\newcommand{\eeqa}{\end{eqnarray}}
\newcommand{\beqan}{\begin{eqnarray*}}
\newcommand{\eeqan}{\end{eqnarray*}}

\newcommand\raiseT[2]{\raisebox{0.25ex}{$#1#2$}
}

\newcommand{\argmax}{\mathop{\rm argmax}}

\newcounter{l1}
\newcounter{l2}
\newcounter{l3}
\newcommand{\bdotlist}{\begin{list}{$\bullet$}{}}
\newcommand{\bboxlist}{\begin{list}{$\Box$}{}}
\newcommand{\bbboxlist}{\begin{list}{\raisebox{.005in}{{\tiny
$\blacksquare$ \ \ }}}{}}
\newcommand{\bdashlist}{\begin{list}{$-$}{} }
\newcommand{\blist}{\begin{list}{}{} }
\newcommand{\barablist}{\begin{list}{\arabic{l1}}{\usecounter{l1}}}
\newcommand{\balphlist}{\begin{list}{(\alph{l2})}{\usecounter{l2}}}
\newcommand{\bAlphlist}{\begin{list}{\Alph{l2}.}{\usecounter{l2}}}
\newcommand{\bdiamlist}{\begin{list}{$\diamond$}{}}
\newcommand{\bromalist}{\begin{list}{(\roman{l3})}{\usecounter{l3}}}

\newtheorem{remark}{Remark}

\renewcommand{\hat}{\widehat}

\renewcommand{\tilde}{\widetilde}
\renewcommand{\top}{{\mathpalette\raiseT\intercal}}

\usepackage{algorithm}
\usepackage{algpseudocode}
\newcommand{\Rmax}{R_{\max}}

\title{Adaptive Agent Design}

\author{Raj Kiriti Velicheti \qquad Subhonmesh Bose \qquad Tamer Ba\c{s}ar
\thanks{All authors are affiliated with the Department of Electrical and Computer
Engineering and the Coordinated Science Laboratory at the University of Illinois,
Urbana-Champaign, Urbana, IL~61801. Emails: \texttt{\{rkv4, boses, basar1\}@illinois.edu}.
This work was partially supported by the U.S. National Science Foundation under grant
number NSF-ECCS-2349418 and by the U.S. Army Research Office under grant number
W911NF-24-1-0085.}}

\begin{document}

\maketitle

\begin{abstract}
We consider an agent acting against a general non-Markovian environment. The agent maintains its agent states, but is free to choose a transition kernel across those states and optimize its state-feedback control policies. We study the bi-level agent design problem that optimizes the transition kernel and the policy it induces, given said kernel with offline data of observations and actions obtained via a behavioral policy. For general environments, we show that a soft $Q$-learning algorithm converges almost surely to the fixed point of a soft Bellman equation defined by the stationary averages that the behavioral policy and the chosen kernel induce, and we delineate what separates the resulting policy from an optimal one. In partially observed Markov decision problems, we analyze convergence properties of parametrized transition kernel design via zero-th order and Bayesian optimization techniques.
\end{abstract}

\section{Introduction}

Intelligent decision-making systems are ubiquitous: from autonomous vehicles navigating
complex traffic to language-model-based agents that invoke tools in open-ended software
environments. Abstracting any such system as an \emph{agent} that interacts with
everything outside it---the \emph{environment}---the agent must assimilate observations
and produce actions that yield favorable outcomes.

A classical approach is to model the environment as a Markov Decision Process (MDP),
explicitly encoding all decision-relevant information into a \emph{state} variable.
MDPs have driven considerable progress in control and reinforcement learning. Their
central challenge, however, is state design: the state must capture \emph{all}
information relevant to future decisions, which is rarely obvious in practice \cite{yuksel2024stochastic}. In
traditional control theory this challenge is addressed through observability and
filtering---from the Kalman filter to particle filters---but these require the
practitioner to specify the state structure in advance. In richer domains such as
language modeling or robotics, the ``right'' state representation is far less transparent.

A more realistic stance is to acknowledge that the environment may be arbitrarily
complex and non-Markovian, and to design agents that learn to summarize their
interaction history into a compact internal state \emph{on the fly}. One principled
formulation of this idea is the \emph{simple agent, complex environment} framework
of~\cite{JMLR:v23:21-0773}, later cast as reinforcement learning in non-Markovian
environments by~\cite{chandak2024reinforcement}. Related work on partially observed
Markov decision processes (POMDPs)~\cite{sinha2024agent,anjarlekarscalable} provides
sample-complexity bounds on the degradation of decision quality as available history
shrinks.

More broadly, these works are related to learning in partially observed environments where the agent is assumed to act without completely knowing the true underlying state. A long line of literature formalizes such settings as POMDPs \cite{aastrom1965optimal}. For discrete spaces, \cite{kaelbling1998planning} formalizes the complexity of finding solutions in these environments. In
 \cite{pineau2003point},  a point based value iteration is proposed to approximately solve these problems in belief space. Another line of work in this direction is utilizing what is known as \emph{Approximate Information State}~\cite{subramanian2022approximate} which proposes a theoretical approach on what a notion of a state should satisfy in partially observed environments that would make it sufficient for acting optimally. See \cite{yuksel2024stochastic} for further literature.

While these works establish a theoretical basis for acting in general environments,
they leave open the question of \emph{how} an agent should adapt its state abstraction
during interaction. In this paper, we address that question directly. We formulate the agent design problem
as an optimization problem over a parametric family of state-transition kernels
$f_\theta$, derive an iterative procedure that alternates between inner policy
optimization and outer adaptation of $\theta$, and study two concrete instantiations of
the outer step: zeroth-order optimization (ZOO) and Bayesian optimization (BO).
We prove that (i) the outer objective $V_D(\theta)$ has a Lipschitz-continuous gradient,
(ii) the ZOO variant converges to a first-order stationary point at rate
$\mathcal{O}(1/\sqrt{K})$, and (iii) the BO variant with a convex-combination
parametrization converges to a global optimum asymptotically.

\section{A Simple Agent in a Non-Markovian Environment}
\label{sec:setup}

Let $\mathcal{A}$ and $\mathcal{O}$ denote finite action and observation spaces,
respectively. The environment is characterized by a stochastic kernel $\rho$ that
generates the next observation
\begin{align*}
    o_{t+1} \;\sim\; \rho\!\left(\cdot \;\middle|\; o^t, a^t\right),
\end{align*}
where $o^t = (o_t,o_{t-1},\dots)$ and $a^t = (a_t, a_{t-1}\dots)$ denote the full histories of
observations and actions up to time $t$. {We assume that the observation process has been evolving for a long enough time i.e., $t>-\infty$. } This formulation allows for \emph{arbitrary}
temporal dependence and does not require the environment to be Markovian.

The agent maintains an internal state $s_t$ taking values in a finite set $\mathcal{S}$.
Rather than assuming the state-update rule to be given or fixed but unknown, we
\emph{parametrize} the update by $\theta \in \Theta \subset \mathbb{R}^d$ and write
\begin{equation}\label{eq:state_update}
  s_{t+1} \;\sim\; f_\theta\!\left(\cdot \;\middle|\; s_t,\, a_t,\, o_{t+1}\right),
\end{equation}
where $f_\theta(\cdot|s,a,o) \in \Delta(\mathcal{S})$ is a stochastic transition kernel.
Equation~\eqref{eq:state_update} is the agent's mechanism for compressing its full
interaction history into a finite state; the quality of this compression depends
entirely on the choice of $\theta$.

Given internal state $s_t$, the agent selects actions according to a policy
$\pi : \mathcal{S} \to \Delta(\mathcal{A})$ and receives reward $r(s_t, a_t, o_{t+1})$,
where we write $|r(\cdot)| \le \Rmax$ for some finite $\Rmax > 0$.
The agent's goal is to maximize the expected discounted cumulative reward. Since this
objective depends on \emph{both} the transition parameter $\theta$ and the policy $\pi$,
we write
\begin{align}\label{eq:objective}
  \max_{\theta \in \Theta,\;\pi}\;
  V(f_\theta, \pi)
  \;=\;
  \mathbb{E}_{\rho, f_\theta, \pi}\left[\sum_{t=0}^{\infty} \gamma^t r(s_t, a_t, o_{t+1})\right],
\end{align}
where $\gamma \in (0,1)$ is a discount factor. Problem~\eqref{eq:objective} is a
\emph{bilevel} optimization problem: at optimality the policy must be optimal \emph{given}
$\theta$. Denoting this best-response policy $\pi^\star_\theta$, substituting it back
yields the outer objective
\begin{align}\label{eq:outer_obj}
  V_D(\theta)
  \;:=\;
  V\!\left(f_\theta,\, \pi^\star_\theta\right)
  \;=\;
  \mathbb{E}_{\rho,\,f_\theta,\,\pi^\star_\theta}\!
  \left[\sum_{t=0}^{\infty}\gamma^t\, r(s_t, a_t, o_{t+1})\right].
\end{align}
Optimizing $V_D(\theta)$ over $\theta$ is the central problem of this paper. We address
it in Section~\ref{sec:algorithm}.

\section{Generic Adaptive Agent Design}
\label{sec:algorithm}

Directly optimizing $V_D(\theta)$ faces two obstacles. First, obtaining
$\pi^\star_\theta$ for a given $f_\theta$ requires solving a reinforcement learning
problem in a non-Markovian environment, for which there is no closed-form solution.
Second, the gradient of $V_D$ with respect to $\theta$ is not analytically accessible
because the environment $\rho$ is treated as a black box.

Our approach resolves both difficulties as follows. Fix a behavioral policy
$\pi_b : \mathcal{S} \to \Delta(\mathcal{A})$ and collect a dataset $D$ of
observation-action-observation triples by rolling $\pi_b$ out in the environment long
enough for the induced process to mix. For a candidate $\theta$, we approximate
$\pi^\star_\theta$ by running soft $Q$-learning directly on the internal states generated
by $f_\theta$, using $D$ as the replay source. Note that this might not result in the true optimal policy because of the non-Markovian environment but is motivated by the fact that the agent ignores the environment complexities and acts simple. The resulting policy $\widehat{\pi}_\theta$
is then deployed in the true environment to obtain a Monte Carlo estimate of
$V_D(\theta)$. Finally, an outer optimization routine updates $\theta$ using only these
function evaluations (no gradient of $V_D$ with respect to $\theta$ is assumed available). This
design is summarized in Algorithm~\ref{alg:main}.

Before stating the algorithm precisely, we record the three assumptions that underlie
the entire development.

\begin{assumption}[Compact and Convex]\label{ass:compact}
$\Theta \subset \mathbb{R}^d$ is compact, convex and $\Pi_\Theta$ denotes Euclidean projection
onto $\Theta$.
\end{assumption}

\begin{assumption}[Smooth internal kernel]\label{ass:smooth_kernel}
For all $\theta \in \Theta$: (i) $f_\theta(s'|s,a,o) > 0$ for all $(s',s,a,o)$; and
(ii) $\theta \mapsto f_\theta$ is twice continuously differentiable with uniformly
bounded first and second derivatives over $\Theta$.
\end{assumption}

\begin{assumption}[Ergodicity and coverage]\label{ass:ergodic}
The behavioral policy $\pi_b$, together with $f_\theta$ and the environment $\rho$,
induces a stationary and ergodic process over $(s_t, a_t, o_{t+1})$ for every
$\theta \in \Theta$. Moreover, every state-action pair is visited:
$\tilde{\pi}(s,a) = P(s_t{=}s, a_t{=}a) > 0$ for all $(s,a) \in \mathcal{S}\times\mathcal{A}$.
\end{assumption}

\begin{algorithm}[t]
\caption{Iterative Agent Design}
\label{alg:main}
\begin{algorithmic}[1]
\Require Behavioral policy $\pi_b$; parametric kernel family $\{f_\theta\}_{\theta\in\Theta}$;
  reward $r$; discount $\gamma\!\in\!(0,1)$; temperature $\tau\!>\!0$;
  outer iterations $K$; evaluation budget $M$.
\Ensure Parameter $\theta_K$ and policy $\widehat\pi_{\theta_K}$.

\State \textbf{Collect data:} Roll out $\pi_b$ in the environment (allowing sufficient
  mixing) to obtain $D = \{(o_{t,i}, a_{t,i}, o_{t+1,i})\}_{i,t}$.
\State Initialize $\theta_0 \in \Theta$.

\For{$k = 0, 1, \dots, K-1$}
  \State \textbf{Policy Computation $\pi^\star_{\theta_k}$:}
    Run soft $Q$-learning with temperature $\tau$ on the internal-state sequence
    generated by $f_{\theta_k}$ using dataset $D$, until convergence to
    $\widehat{Q}_{\theta_k}$; set
    $\widehat\pi_{\theta_k}(a|s) \propto \exp(\widehat{Q}_{\theta_k}(s,a)/\tau)$.

  \State \textbf{Policy Evaluation:}
    Estimate $V_D(\theta_k)$ by rolling out $(f_{\theta_k}, \widehat\pi_{\theta_k})$
    in the true environment for $M$ episodes of horizon $H$.

  \State \textbf{Parameter Update:}
    $\displaystyle\theta_{k+1}
      = \textsc{StepOptimize}\!\left(
          V_D,\;\theta_k,\;
          \bigl\{(\theta_j,\, V_D(\theta_j))\bigr\}_{j \le k}
        \right)$
\EndFor

\State \Return $\theta_K$ and $\widehat\pi_{\theta_K}$.
\end{algorithmic}
\end{algorithm}

Two concrete implementations of \textsc{StepOptimize} are discussed below and
detailed in Algorithms~\ref{alg:zoo} and~\ref{alg:bo}.

\subsection{StepOptimize I: Zeroth-Order Optimization}
\label{sec:zoo}

When the gradient $\nabla_\theta V_D(\theta)$ is unavailable, we estimate it via a
two-point finite-difference scheme. At iteration $k$, draw a unit direction
$u_k \sim \mathrm{Uniform}(\mathbb{S}^{d-1})$, evaluate $V_D$ at $\theta_k \pm \delta_k u_k$,
and form the gradient estimate
\begin{equation}\label{eq:zo_grad}
  \hat{g}_k
  = \frac{d}{2\delta_k}\left(\widehat{V}_D(\theta_k + \delta_k u_k)
          - \widehat{V}_D(\theta_k - \delta_k u_k)\right)u_k.
\end{equation}
The iterate is updated as
$\theta_{k+1} = \Pi_\Theta(\theta_k + \eta_k \hat{g}_k)$,
where $\Pi_\Theta$ is Euclidean projection onto $\Theta$ and $\eta_k, \delta_k > 0$
are the step size and smoothing radius, respectively.

\begin{algorithm}[t]
\caption{\textsc{StepOptimize} via Zeroth-Order Optimization (ZOO)}
\label{alg:zoo}
\begin{algorithmic}[1]
\Require Current iterate $\theta_k$; objective oracle $V_D$; step size $\eta_k$;
  smoothing radius $\delta_k$; evaluation budget $M$.
\Ensure Updated parameter $\theta_{k+1}$.
\State Draw $u_k \sim \mathrm{Uniform}(\mathbb{S}^{d-1})$.
\State Estimate $\widehat{V}_+$ and $\widehat{V}_-$ by rolling out
  $(f_{\theta_k \pm \delta_k u_k},\, \widehat\pi_{\theta_k \pm \delta_k u_k})$
  for $M$ episodes each.
\State $\hat{g}_k = \dfrac{d}{2\delta_k}\left(\widehat{V}_+ - \widehat{V}_-\right)u_k$.
\State $\theta_{k+1} = \Pi_\Theta\!\left(\theta_k + \eta_k\, \hat{g}_k\right)$.
\State \Return $\theta_{k+1}$.
\end{algorithmic}
\end{algorithm}

\subsection{StepOptimize II: Bayesian Optimization}
\label{sec:bo}

When each evaluation of $V_D$ is expensive, Bayesian optimization (BO) is preferable:
it maintains a Gaussian Process (GP) surrogate fitted to all past evaluations and
selects the next query by maximizing an acquisition function, thus exploiting global
smoothness with far fewer function evaluations.

We place a GP prior on $V_D$ with a Mat\'{e}rn-$5/2$ covariance in the Euclidean metric
on $\Theta$. Writing $r = \|\theta - \theta'\|_2$,
\begin{equation*}
  \kappa(\theta,\theta')
  = \sigma_f^2
    \!\left(1 + \frac{\sqrt{5}\,r}{\ell} + \frac{5\,r^2}{3\ell^2}\right)
    \exp\!\left(-\frac{\sqrt{5}\,r}{\ell}\right),
\end{equation*}
with signal variance $\sigma_f^2$ and length scale $\ell$. Under the convex-combination
parametrization $f_\theta = \sum_{i=1}^N \theta_i f_i$ of
Lemma~\ref{lem:convex_kernel}, the domain $\Theta$ is the $(N-1)$-simplex, a compact
convex subset of $\mathbb{R}^N$ on which $\kappa$ is strictly positive definite. It is
the smoothness this kernel encodes, through $\nu = 5/2$, that
Theorem~\ref{thm:bo_convergence} exploits. At each iteration the next candidate is chosen
by maximizing the Upper Confidence Bound (UCB) acquisition:
$\theta_{k+1} = \argmax_{\theta\in\Theta}
  \bigl(\mu_k(\theta) + \beta_{k+1}\,\sigma_k(\theta)\bigr)$.

\begin{algorithm}[t]
\caption{\textsc{StepOptimize} via Bayesian Optimization (BO)}
\label{alg:bo}
\begin{algorithmic}[1]
\Require History $\mathcal{H}_{k} = \{(\theta_j,\widehat{V}_D(\theta_j))\}_{j \le k}$;
  kernel $\kappa$; observation noise variance $\sigma_\xi^2$;
  acquisition parameter $\beta_{k+1}$; evaluation budget $M$.
\Ensure Updated parameter $\theta_{k+1}$.
\State Let $\mathbf{y}_k = [\widehat{V}_D(\theta_j)]_{j \le k}$ and
  $[K_k]_{ij} = \kappa(\theta_i,\theta_j)$.
\State Compute $W_k = (K_k + \sigma_\xi^2 I)^{-1}$.
\State For each candidate $\theta$, let
  $\mathbf{k}_k(\theta) = [\kappa(\theta,\theta_j)]_{j\le k}$ and compute
  \vspace{-2pt}
  \begin{align*}
    \mu_k(\theta) &= \mathbf{k}_k(\theta)^\top W_k\,\mathbf{y}_k,\\
    \sigma_k(\theta) &= \sqrt{\kappa(\theta,\theta)
      - \mathbf{k}_k(\theta)^\top W_k\,\mathbf{k}_k(\theta)}\,.
  \end{align*}
\State $\theta_{k+1} = \argmax_{\theta \in \Theta}
  \bigl(\mu_k(\theta) + \beta_{k+1}\,\sigma_k(\theta)\bigr)$.
\State Estimate $\widehat{V}_D(\theta_{k+1})$ using $M$ rollout episodes.
\State Append $(\theta_{k+1}, \widehat{V}_D(\theta_{k+1}))$ to $\mathcal{H}_{k+1}$.
\State \Return $\theta_{k+1}$.
\end{algorithmic}
\end{algorithm}

\section{Convergence of Inner Soft $Q$-Learning}
\label{sec:inner}

Before analyzing the outer loop, we establish that the inner step of
Algorithm~\ref{alg:main} is well-founded: soft $Q$-learning converges almost surely
to a well-defined fixed point even though the internal state process is
\emph{not} Markovian (because $\rho$ has arbitrary memory).
The key insight, formalized below, is that the non-Markovian noise term vanishes in
expectation under the stationary measure induced by $\pi_b$, and thus the iteration
effectively tracks a Markovian ODE.

\begin{theorem}[Almost-sure convergence of soft $Q$-learning]
\label{thm:softq_convergence}
Let Assumption~\ref{ass:ergodic} hold and let the learning
rates $\{\alpha_n\}$ satisfy the Robbins--Monro conditions
$\sum_n \alpha_n = \infty$ and $\sum_n \alpha_n^2 < \infty$.
Define the stationary expected reward and transition induced by $\pi_b$ and $f_\theta$:
\begin{align}
  \bar{r}_\theta(s,a)
  &= \mathbb{E}_{\pi_b, f_\theta}\!\left[r(s,a,o)\right],
  \label{eq:rbar}\\
  \bar{P}_\theta(s'|s,a)
  &= \mathbb{E}_{\pi_b, f_\theta}\!\left[f_\theta(s'|s,a,o)\right].
  \label{eq:Pbar}
\end{align}
Then, the soft $Q$-learning iterates
\begin{align}\label{eq:softq_update}
\begin{split}
  Q_{n+1}(s,a)
  \;&=\; Q_n(s,a)
    + \alpha_n\, \mathbf{1}_{\{S_n=s,\,A_n=a\}}
    \!\Big[r_n\\
      &+ \gamma\,\tau\log\!\sum_{a'}\exp\!\tfrac{Q_n(S_{n+1},a')}{\tau}
      - Q_n(s,a)\Big]
\end{split}
\end{align}
converge almost surely to the unique fixed point $Q^\star_\theta$ satisfying
\begin{equation}\label{eq:soft_bellman}
  Q^\star_\theta(s,a)
  = \bar{r}_\theta(s,a)
    + \gamma\sum_{s'}\bar{P}_\theta(s'|s,a)\,
      \tau\log\!\sum_{a'}\exp\!\tfrac{Q^\star_\theta(s',a')}{\tau}.
\end{equation}
The induced soft-greedy policy is
$\widehat\pi_\theta(a|s) \propto \exp(Q^\star_\theta(s,a)/\tau)$.
\end{theorem}

\begin{proof}
Decompose the update in ~\eqref{eq:softq_update} as
\begin{equation*}
  Q_{n+1}(s,a)
  = Q_n(s,a) + \alpha_n\!
    \left[F^{s,a}(Q_n) + \zeta^{s,a}_n + M^{s,a}_{n+1}\right],
\end{equation*}
where $F^{s,a}(Q)$ is the expected soft Bellman residual under
$(\bar{r}_\theta, \bar{P}_\theta)$:
\begin{align*}
     F^{s,a}(Q)
  &= \mathbf{1}_{\{S_n=s,A_n=a\}}
    \!\Big[
      \bar{r}_\theta(s,a)
      +\\ &\gamma\sum_{s'}\bar{P}_\theta(s'|s,a)\,\tau\log\!\sum_{a'}e^{Q(s',a')/\tau}
      - Q(s,a)
    \Big]; 
\end{align*}

$\zeta^{s,a}_n$ captures the non-Markovian correction
\begin{align*}
    \zeta^{s,a}_n
  &= \mathbf{1}_{\{S_n=s,A_n=a\}}
    \Big[
      \mathbb{E}[r_n | o^n, a^n] -\\
      & \bar{r}_\theta(s,a)
      + \gamma\!\sum_{s'}\!
        \Bigl(P(S_{n+1}{=}s'|o^n,a^n)
          - \\
          &\bar{P}_\theta(s'|s,a)\Bigr)
        \tau\log\!\sum_{a'}e^{Q_n(s',a')/\tau}
    \Big];
\end{align*}
and $M^{s,a}_{n+1}$ is a martingale difference sequence with bounded increments (since
rewards are bounded by $\Rmax$).

Under the stationary measure $\tilde\pi(s,a) > 0$
(Assumption~\ref{ass:ergodic}), integrating $\zeta^{s,a}_n$ over the history $y=(o^n,a^n)$
gives zero:
\begin{equation*}
  \int \tilde\pi(dy)\,\zeta^{s,a}(Q,y) = 0
  \quad \forall\,(s,a),
\end{equation*}
because $\bar{P}_\theta(s'|s,a)
= \mathbb{E}_{\tilde\pi}[P(S_{n+1}{=}s'|o^n,a^n) \mid S_n{=}s, A_n{=}a]$
by definition of the stationary averages in~\eqref{eq:Pbar}.
Hence, the non-Markovian noise vanishes in expectation under the stationary measure.
Standard stochastic approximation theory~\cite{borkar2008stochastic} then implies that
$\{Q_n\}$ almost surely tracks the ODE
$\dot Q = \sum_{s,a}\tilde\pi(s,a)\,F^{s,a}(Q)$.

The operator underlying $F^{s,a}$ is the soft Bellman operator with temperature $\tau>0$,
which is a $\gamma$-contraction in $\ell_\infty$; see, e.g.,~\cite{haarnoja2018soft}.
Hence the ODE is globally asymptotically stable at the unique fixed point $Q^\star_\theta$
of~\eqref{eq:soft_bellman}, giving $Q_n \xrightarrow{a.s.} Q^\star_\theta$.
\end{proof}

\begin{remark}\label{rem:nearopt}
Theorem~\ref{thm:softq_convergence} is a statement about the surrogate chain
$(\bar{r}_\theta, \bar{P}_\theta)$: the iteration converges, and $\widehat{\pi}_\theta$ is
soft-optimal for that chain. It does not by itself certify that $\widehat{\pi}_\theta$ is
near-optimal for the environment. Two gaps separate the
two. The first is the temperature, which biases the soft-optimal policy away from the
greedy one by at most $\tau\log|\mathcal{A}|/(1-\gamma)$ in value; this gap is under the
designer's control and is the price paid for the smoothness of
$\theta \mapsto \widehat{\pi}_\theta$ that Section~\ref{sec:theory} requires. The second is
representational, and is the essential one: $\widehat{\pi}_\theta$ acts on $s_t$ rather than
on the history, so it can be no better than the best policy the compression $f_\theta$
admits, and $(\bar{r}_\theta,\bar{P}_\theta)$ are themselves averages taken under $\pi_b$.
Bounding this second gap requires structure that a general non-Markovian environment does
not provide. When the environment is a POMDP and $f_\theta$ is a filter whose prediction
error contracts, filter-stability arguments bound the loss in value by the expected
filtering error accumulated over the effective horizon
$1/(1-\gamma)$~\cite{kara2023convergence, subramanian2022approximate}, and in that regime
$\widehat{\pi}_\theta$ is near-optimal in the usual sense.

This also locates the result relative to prior work. The closest antecedent
is~\cite{kara2023convergence}, which fixes a finite-window memory as the agent state and
establishes both convergence and, under filter stability, near-optimality of the learned
policy. Here the agent state is instead the output of a parametrized, \emph{designable}
kernel $f_\theta$, so the surrogate chain and hence the limit $Q^\star_\theta$ move with
$\theta$, and the argument must hold uniformly over the family
$\{f_\theta\}_{\theta\in\Theta}$; what the proof isolates is that the non-Markovian
correction $\zeta^{s,a}_n$ integrates to zero under the stationary measure for
\emph{every} $\theta$, and it is this uniformity that makes the outer problem well posed.
It is precisely because the representational gap depends on $\theta$, and is untouched by
the inner loop, that we treat the choice of $\theta$ as an optimization problem in its own
right.
\end{remark}

Theorem~\ref{thm:softq_convergence} establishes that the inner step of
Algorithm~\ref{alg:main} converges almost surely to a well-defined fixed point
$Q^\star_\theta$, yielding a policy $\widehat{\pi}_\theta$ with which the agent
can act in the non-Markovian environment for any fixed $\theta$. However, this
addresses only the inner problem, and the outer question (\emph{which $\theta$ to
use}) is still to be addressed. While Algorithm~\ref{alg:main} provides a general iterative
procedure for adapting $\theta$, without further structure on the environment
it is difficult to characterize the quality of the stationary points it reaches.
In what follows, we restrict the environment to be a POMDP
(Assumption~\ref{ass:pomdp}) and analyze the properties of the transition
parameters obtained by running Algorithm~\ref{alg:main} in this setting. 
\section{Theoretical Guarantees under a POMDP Environment}
\label{sec:theory}

The convergence results in Section~\ref{sec:inner} hold for an arbitrary non-Markovian
environment. To prove any structural properties for the stationary points of $V_D(\theta)$, such as Lipschitzness of the gradient, we need additional
structure on the environment. We now assume that it is a \emph{Partially Observed Markov
Decision Process} (POMDP), which captures most practically relevant settings while
enabling sharper analysis.

\begin{assumption}[POMDP environment]\label{ass:pomdp}
There exist a finite latent state space $\mathcal{X}$, a transition kernel $P(x'|x,a)$,
and an emission kernel $O(o|x)$ such that
\begin{equation*}
  x_{t+1} \sim P(\cdot \mid x_t, a_t),
  \qquad
  o_{t+1} \sim O(\cdot \mid x_{t+1}),
\end{equation*}
with the agent's internal state evolving via
$s_{t+1} \sim f_\theta(\cdot \mid s_t, a_t, o_{t+1})$.
\end{assumption}

Under Assumption~\ref{ass:pomdp}, the joint process $(x_t, s_t)$ is Markovian, even
though the observations and internal states alone are not. This latent Markov structure
is the key that allows us to express $V_D(\theta)$ as the value of a well-defined joint
Markov chain and then differentiate it.

\subsection{Lipschitz Gradient of the Outer Objective}
\label{sec:lipschitz}

Our goal is to show that $\nabla_\theta V_D(\theta)$ is Lipschitz continuous, which is
the regularity condition needed for both convergence proofs that follow. Establishing
this requires bounding the first and second derivatives of $V_D$ with respect to
$\theta$, which we build up in two lemmas: one for the quantities
$(\bar{r}_\theta, \bar{P}_\theta)$ defined by the behavioral data, and another one for the
closed-loop joint model under the optimized policy.

\begin{lemma}[Smooth stationary averages]\label{lem:surrogate_bounds}
Under Assumptions~\ref{ass:smooth_kernel},~\ref{ass:ergodic},
and~\ref{ass:pomdp}, there exist finite constants
$C_{\bar{r},1}, C_{\bar{r},2}, C_{\bar{P},1}, C_{\bar{P},2}$ such that for all
$s, s' \in \mathcal{S}$, $a \in \mathcal{A}$, and $\theta \in \Theta$:
\begin{align*}
  \|\nabla_\theta \bar{r}_\theta(s,a)\|_2 &\le C_{\bar{r},1},
  &\|\nabla^2_\theta \bar{r}_\theta(s,a)\|_2 &\le C_{\bar{r},2},\\
  \|\nabla_\theta \bar{P}_\theta(s'|s,a)\|_2 &\le C_{\bar{P},1},
  &\|\nabla^2_\theta \bar{P}_\theta(s'|s,a)\|_2 &\le C_{\bar{P},2}.
\end{align*}
\end{lemma}

\begin{proof}
Let $\mathcal{Z} = \mathcal{X} \times \mathcal{S}$ and $z = (x,s)$. Under $\pi_b$,
$f_\theta$, and the POMDP dynamics $(P, O)$, the joint process on $\mathcal{Z}$ is a
Markov chain with transition kernel
\begin{equation*}
  P_{\mathrm{jt}}(z'|z)
  = \sum_{a} \pi_b(a|s)\, P(x'|x,a) \sum_{o} O(o|x')\, f_\theta(s'|s,a,o).
\end{equation*}
By Assumption~\ref{ass:ergodic} this chain is ergodic with unique stationary
distribution $d_\theta$, satisfying $d_\theta^\top = d_\theta^\top P_{\mathrm{jt}}$.
Let $Z_\theta = (I - P_{\mathrm{jt}} + \mathbf{1} d_\theta^\top)^{-1}$ be the fundamental
matrix. Differentiating the stationarity equation gives
\begin{equation*}
  \nabla_\theta d_\theta^\top = d_\theta^\top (\nabla_\theta P_{\mathrm{jt}}) Z_\theta.
\end{equation*}
By Assumption~\ref{ass:smooth_kernel}, $\|\nabla_\theta f_\theta\|_2 \le C_{f,1}$, and thus
$\|\nabla_\theta P_{\mathrm{jt}}\|_2$ is bounded. Since $\mathcal{Z}$ is finite and the
chain is ergodic, $\|Z_\theta\|$ is bounded by a constant proportional to the mixing
time, yielding a constant $C_{d,1}$ with
$\sum_{x}\|\nabla_\theta d_\theta(x,s)\|_2 \le C_{d,1}$ for all $s$.

Define the conditional observation weight
$w_\theta(o|s,a) = N_\theta(o,s,a) / D_\theta(s)$, where
$D_\theta(s) = \sum_x d_\theta(x,s)$ and
$N_\theta(o,s,a) = \sum_x d_\theta(x,s) \sum_{x'} P(x'|x,a) O(o|x')$.
Differentiating the numerator and the denominator separately gives
$\|\nabla_\theta D_\theta(s)\|_2 \le C_{d,1}$ and
$\|\nabla_\theta N_\theta(o,s,a)\|_2 \le C_{d,1}$.
Applying the quotient rule with $D_\theta(s) \ge \mu_{\min}$ (Assumption~\ref{ass:ergodic}):
\begin{equation*}
  \|\nabla_\theta w_\theta(o|s,a)\|_2
  \;\le\; \frac{2\,C_{d,1}}{\mu_{\min}}
  \;\triangleq\; C_{w,1}.
\end{equation*}

Since $\bar{r}_\theta(s,a) = \sum_o w_\theta(o|s,a)\,r(s,a,o)$, and $|r(\cdot)|\le\Rmax$:
\begin{equation*}
  \|\nabla_\theta \bar{r}_\theta(s,a)\|_2
  \;\le\; |\mathcal{O}|\,C_{w,1}\,\Rmax
  \;\triangleq\; C_{\bar{r},1}.
\end{equation*}
Since $\bar{P}_\theta(s'|s,a) = \sum_o w_\theta(o|s,a)\,f_\theta(s'|s,a,o)$, the
product rule and the fact that $\sum_o f_\theta(s'|s,a,o)$ and $\sum_o w_\theta(o|s,a)$
each equal $1$ lead to
\begin{equation*}
  \|\nabla_\theta \bar{P}_\theta(s'|s,a)\|_2
  \;\le\; C_{w,1} + C_{f,1}
  \;\triangleq\; C_{\bar{P},1}.
\end{equation*}

The second-derivative bounds $C_{\bar{r},2}$ and $C_{\bar{P},2}$ follow by
differentiating each expression above once more. Assumption~\ref{ass:smooth_kernel}
supplies $\|\nabla^2_\theta f_\theta\|_2 \le C_{f,2}$, which propagates through
$P_{\mathrm{jt}}$, $d_\theta$, $w_\theta$, and finally $\bar{r}_\theta$ and
$\bar{P}_\theta$ by iterated application of the product and quotient rules.
\end{proof}

Lemma~\ref{lem:surrogate_bounds} controls how the stationary averages $\bar r_\theta$
and $\bar P_\theta$ change as $\theta$ varies. The next lemma lifts this to the
\emph{closed-loop} joint model---the reward and transition seen by the joint process
$(x_t, s_t)$ when the agent follows the optimized policy $\widehat\pi_\theta$.
The key step is propagating the $\theta$-dependence of $\widehat\pi_\theta$ through
the soft Bellman equation.

\begin{lemma}[Smooth closed-loop model]\label{lem:joint_bounds}
Define the joint reward and transition under the
closed-loop system $(f_\theta, \widehat\pi_\theta)$:
\begin{align*}
  r^\theta(z)
  &= \sum_{a,o,x'} \widehat\pi_\theta(a|s)\,P(x'|x,a)\,O(o|x')\,r(s,a,o),\\
  P^\theta(z'|z)
  &= \sum_{a,o} \widehat\pi_\theta(a|s)\,P(x'|x,a)\,O(o|x')\,f_\theta(s'|s,a,o).
\end{align*}
Under Assumptions~\ref{ass:compact}--\ref{ass:pomdp}, there exist finite constants $C_{r,1}, C_{r,2}, C_{P,1}, C_{P,2}$ such that for all
$z, z' \in \mathcal{Z}$ and $\theta \in \Theta$:
\begin{equation*}
  \|\nabla_\theta r^\theta(z)\|_2 \le C_{r,1},
  \quad
  \|\nabla_\theta P^\theta(z'|z)\|_2 \le C_{P,1},
\end{equation*}
and the same bounds hold for second derivatives with constants $C_{r,2}$ and $C_{P,2}$.
\end{lemma}

\begin{proof}
Since $P$ and $O$ are fixed (independent of $\theta$), bounding
$\|\nabla_\theta r^\theta\|$ and $\|\nabla_\theta P^\theta\|$ reduces to bounding
$\|\nabla_\theta \widehat\pi_\theta(a|s)\|_2$.

Since $\widehat\pi_\theta(a|s) \propto \exp(Q^\star_\theta(s,a)/\tau)$ (with
temperature $\tau > 0$), the chain rule through the softmax gives
\begin{equation*}
  \|\nabla_\theta \widehat\pi_\theta(a|s)\|_2
  \;\le\; \frac{2}{\tau}\,B_Q,
  \qquad
  B_Q \;\triangleq\; \max_{s,a}\|\nabla_\theta Q^\star_\theta(s,a)\|_2.
\end{equation*}
To bound $B_Q$, differentiate the soft Bellman equation~\eqref{eq:soft_bellman}
implicitly with respect to $\theta$:
\begin{align*}
  \nabla_\theta Q^\star_\theta(s,a)
  &= \nabla_\theta \bar{r}_\theta(s,a)
    + \gamma\sum_{s'}\!\Bigl[
      \nabla_\theta \bar{P}_\theta(s'|s,a)\,V^\star_\theta(s')\\
    &\qquad + \bar{P}_\theta(s'|s,a)
        \sum_{a'}\widehat\pi_\theta(a'|s')\,
        \nabla_\theta Q^\star_\theta(s',a')
    \Bigr],
\end{align*}
where $V^\star_\theta(s') = \tau\log\sum_{a'}\exp(Q^\star_\theta(s',a')/\tau)$.
Taking the $\ell_2$ norm, using $|V^\star_\theta(s')| \le \Rmax/(1-\gamma)$, and
exploiting the contraction factor $\gamma < 1$, yields the self-consistent bound
\begin{equation*}
  B_Q
  \;\le\;
  \frac{1}{1-\gamma}
  \!\left(
    C_{\bar{r},1}
    + \frac{\gamma \Rmax\,|\mathcal{S}|}{1-\gamma}\,C_{\bar{P},1}
  \right)
  \;\triangleq\; C_Q,
\end{equation*}
which is finite by Lemma~\ref{lem:surrogate_bounds}. This gives constants $C_{r,1}$ and
$C_{P,1}$.

The second-derivative bounds $C_{r,2}$ and $C_{P,2}$ follow by differentiating the soft
Bellman equation once more with respect to $\theta$. The resulting expression for
$\nabla^2_\theta Q^\star_\theta$ involves $\nabla^2_\theta \bar{r}_\theta$,
$\nabla^2_\theta \bar{P}_\theta$, $\nabla_\theta \widehat\pi_\theta$, and
$\nabla_\theta Q^\star_\theta$---all of which are already bounded---together with the
same $\gamma$-contraction argument, yielding a finite constant $B_{Q,2}$. The bounds
$C_{r,2}$ and $C_{P,2}$ then follow by substitution.
\end{proof}

With both lemmas in hand, we can now establish the key structural property of $V_D$.

\begin{theorem}[Lipschitz gradient of $V_D$]\label{thm:lipschitz}
Under Assumptions~\ref{ass:compact}--\ref{ass:pomdp}, the gradient
$\nabla_\theta V_D(\theta)$ is Lipschitz continuous on $\Theta$.
\end{theorem}

\begin{proof}
Under Assumption~\ref{ass:pomdp}, define the joint state $z_t = (x_t, s_t) \in \mathcal{Z}$.
Under the closed-loop pair $(f_\theta, \widehat\pi_\theta)$, the joint process is a
Markov chain with transition $P^\theta$ and one-step reward $r^\theta$ from
Lemma~\ref{lem:joint_bounds}. The value function $V^\theta \in \mathbb{R}^{|\mathcal{Z}|}$
satisfies the Bellman equation
\begin{equation*}
  V^\theta = r^\theta + \gamma\, P^\theta V^\theta
  \;\;\Longrightarrow\;\;
  V^\theta = (I - \gamma P^\theta)^{-1} r^\theta,
\end{equation*}
and $V_D(\theta) = \mu_0^\top V^\theta$ for some initial distribution $\mu_0 \in
\Delta(\mathcal{Z})$. Since $\|V^\theta\|_\infty \le \Rmax/(1-\gamma)$, differentiating
the Bellman equation gives
\begin{align*}
  \nabla_\theta V^\theta
  = (I - \gamma P^\theta)^{-1}
    \!\left(\nabla_\theta r^\theta + \gamma\,(\nabla_\theta P^\theta)\, V^\theta\right).
\end{align*}
Since $\|(I-\gamma P^\theta)^{-1}\| \le 1/(1-\gamma)$, Lemma~\ref{lem:joint_bounds}
yields
\begin{equation*}
  \|\nabla_\theta V^\theta\|
  \;\le\;
  \frac{C_{r,1} + \gamma\,\Rmax\,C_{P,1}/(1-\gamma)}{1-\gamma}
  \;\triangleq\; C_{V,1}.
\end{equation*}
Differentiating the Bellman equation once more:
\begin{align*}
  \nabla^2_\theta V^\theta
  &= (I-\gamma P^\theta)^{-1}\!
    \Big(
      \nabla^2_\theta r^\theta
      + \gamma\,(\nabla^2_\theta P^\theta)\, V^\theta
      \\
      &+ 2\gamma\,(\nabla_\theta P^\theta)\,(\nabla_\theta V^\theta)
    \Big).
\end{align*}
Applying the second-derivative bounds from Lemma~\ref{lem:joint_bounds}:
\begin{align}\label{eq: lipschitz_bound}
  \|\nabla^2_\theta V^\theta\|
  \leq
  \frac{C_{r,2} + \gamma\,\Rmax\,C_{P,2}/(1-\gamma) + 2\gamma\,C_{P,1}\,C_{V,1}}{1-\gamma}.
\end{align}
Since $\mu_0$ is a probability vector, defining the bound in \ref{eq: lipschitz_bound} as $L$, we have
$\|\nabla^2_\theta V_D(\theta)\| \le L$ uniformly over $\Theta$, which implies that
$\nabla_\theta V_D$ is $L$-Lipschitz continuous.
\end{proof}

\subsection{Bounded Variance of the Monte Carlo Estimator}
\label{sec:mc_variance}

In Algorithm~\ref{alg:main} we estimate $V_D(\theta)$ by averaging $M$ independent
rollouts of horizon $H$. Since rewards are bounded, the truncated return has bounded
variance and the estimator is sub-Gaussian---a property used in both convergence
analyses below.

\begin{lemma}[Bounded-variance Monte Carlo estimator]\label{lem:mc_variance}
For a fixed $\theta$, let $\widehat{V}_D(\theta) = \frac{1}{M}\sum_{m=1}^M G_m$ where
$G_m = \sum_{t=0}^{H-1}\gamma^t r_m(s_t,a_t,o_{t+1})$ is the discounted return of the
$m$-th rollout. Then:
\begin{enumerate}
  \item[\normalfont(i)] each $G_m$ is bounded, with
    $|G_m| \le G_{\max} \triangleq \Rmax(1-\gamma^H)/(1-\gamma)$;
  \item[\normalfont(ii)] $\widehat{V}_D(\theta)$ is an unbiased estimator of the
    $H$-truncated return, with conditional variance
    $\mathrm{Var}[\widehat{V}_D(\theta)] \le G_{\max}^2/M \le \sigma^2/M$, where
    $\sigma^2 \triangleq \Rmax^2/(1-\gamma)^2$; and
  \item[\normalfont(iii)] $\widehat{V}_D(\theta)$ is $\sigma/\sqrt{M}$-sub-Gaussian, i.e.\
    $\log \mathbb{E}\bigl[\exp\bigl(\lambda(\widehat{V}_D(\theta) - \mathbb{E}\widehat{V}_D(\theta))\bigr)\bigr]
    \le \lambda^2\sigma^2/(2M)$ for all $\lambda \in \mathbb{R}$.
\end{enumerate}
\end{lemma}

\begin{proof}
Part (i) is immediate from $|r(\cdot)|\le\Rmax$ and the geometric series. For part (ii),
the $M$ rollouts are independent, so
$\mathrm{Var}[\widehat{V}_D(\theta)] = \mathrm{Var}[G_1]/M$, and $G_1$ being bounded by
$G_{\max}$, Popoviciu's inequality gives $\mathrm{Var}[G_1] \le G_{\max}^2$. Part (iii)
follows from Hoeffding's lemma for bounded random variables: each $G_m/M$ takes values in
an interval of length $2G_{\max}/M$, so the average is $\sigma/\sqrt{M}$-sub-Gaussian.

\end{proof}

Lemma~\ref{lem:mc_variance} makes precise how the evaluation budget $M$ and the
discount factor $\gamma$ jointly determine the estimation noise. In particular,
$\sigma^2 = \Rmax^2/(1-\gamma)^2$ grows as $\gamma \to 1$---reflecting the increasing
variance of long-horizon returns---and shrinks at rate $1/M$ as more rollouts are used.
The sub-Gaussian property (part iii) is what allows us to invoke GP-UCB confidence
bounds in the Bayesian optimization analysis.

\subsection{Convergence of Algorithm~\ref{alg:main} with ZOO}
\label{sec:zoo_convergence}

With the Lipschitz gradient established and the estimator noise characterized, we can
now analyze the zeroth-order variant. The two-point estimator~\eqref{eq:zo_grad}
introduces two sources of error: a \emph{bias} from the finite smoothing radius
$\delta$ and a \emph{variance} from the Monte Carlo evaluations. The theorem below
shows that both can be controlled, and balancing the two gives an
$\mathcal{O}(1/\sqrt{K})$ convergence rate to a first-order stationary point of $V_D$.

\begin{theorem}[Convergence of ZOO variant]
\label{thm:zoo_convergence}
Under Assumptions~\ref{ass:compact}--\ref{ass:pomdp}, Algorithm~\ref{alg:main} with Algorithm~\ref{alg:zoo} as \textsc{StepOptimize}
for $K$ iterations with constant step size $\eta \le 1/(2L)$ and constant smoothing
radius $\delta > 0$ yields,
\begin{align}
  \frac{1}{K}\sum_{k=0}^{K-1}\mathbb{E}\left[\|G_\eta(\theta_k)\|_2^2\right]
  \leq
  \frac{C_1}{\eta K}
  + \frac{2 L d^2\sigma^2}{M\delta^2}
  + 2\delta^2 d^2 L^2,
\end{align}
where $G_\eta(\theta) = \frac{1}{\eta}(\theta - \Pi_\Theta(\theta +
\eta\,\nabla V_D(\theta)))$ is the gradient mapping, $C_1 = 2\bigl(V^\star - V_D(\theta_0)\bigr)$, 
$V^\star = \max_{\theta\in\Theta} V_D(\theta)$, and $d = \dim\Theta$.
Choosing $\eta = \mathcal{O}(K^{-1/2})$ and $\delta = \mathcal{O}(K^{-1/4})$
gives $\frac{1}{K}\sum_k \mathbb{E}[\|G_\eta(\theta_k)\|_2^2] =
\mathcal{O}(K^{-1/2})$. 
\end{theorem}
\begin{proof}
Define the spherically smoothed objective
$V_\delta(\theta) = \mathbb{E}_{u \sim \mathrm{Uniform}(\mathbb{S}^{d-1})}[V_D(\theta + \delta u)]$.
By the Stokes' theorem identity for uniform spherical smoothing~\cite{flaxman2004online},
the two-point estimator satisfies
$\mathbb{E}[\hat{g}_k \mid \theta_k] = \nabla V_\delta(\theta_k)$.
Due to $L$-Lipschitz gradient (Theorem~\ref{thm:lipschitz}), using standard techniques from gradient free optimization
\cite{duchi2015optimal} and splitting $\hat{g}_k$ into signal
and noise and applying $(a+b)^2 \leq 2a^2 + 2b^2$, the bias and second
moment bounds for the uniform spherical estimator are given by:
\begin{align*}
  &\|\nabla V_\delta(\theta) - \nabla V_D(\theta)\|_2 \leq \frac{Ld\delta}{2},\\
  \mathbb{E}[\|\hat{g}_k\|_2^2 \mid \theta_k]
  &\leq 4d^2\|\nabla V_D(\theta_k)\|_2^2
    + \frac{d^2\sigma^2}{M\delta^2}
    + \delta^2 d^2 L^2,
\end{align*}
where $\sigma^2 = \Rmax^2/(1-\gamma)^2$ (Lemma~\ref{lem:mc_variance}) accounts
for the Monte Carlo evaluation noise.

Applying the projected gradient ascent analysis of~\cite{ghadimi2013stochastic}
to the $L$-smooth objective $V_D$ over the convex compact set $\Theta$, with
$G_\eta(\theta) = \frac{1}{\eta}(\theta - \Pi_\Theta(\theta + \eta\nabla V_D(\theta)))$
as the stationarity measure, telescoping over $k = 0,\dots,K-1$ yields the
stated bound. Setting $\eta = c_1 K^{-1/2}$ and $\delta = c_2 K^{-1/4}$ balances
all three terms at $\mathcal{O}(K^{-1/2})$.
\end{proof}

\subsection{Convergence of Algorithm~\ref{alg:main} with Bayesian Optimization}
\label{sec:bo_convergence}

The ZOO variant makes no assumption on the global structure of $V_D$ and consequently
converges only to a first-order stationary point. By contrast, the BO variant exploits
the \emph{smoothness} of $V_D$, captured by its membership in a Reproducing Kernel
Hilbert Space, to converge to the \emph{global} optimum. This comes at the cost of a
stronger parametrization requirement (for example, convex combination of base kernels) and the
sub-Gaussian noise model, both of which we make precise below.

For the convex-combination parameterization, the following lemma verifies that the
structural assumptions of the paper hold automatically, and moreover that $V_D$ is
infinitely differentiable, the regularity needed for RKHS membership.

\begin{lemma}[Convex combination kernel]\label{lem:convex_kernel}
Let $f_\theta = \sum_{i=1}^N \theta_i f_i$ with $\theta\in\Theta$ the standard
$(N-1)$-simplex, and suppose each base kernel $f_i > 0$ induces an ergodic joint chain
on $\mathcal{Z}$. Then:
\begin{enumerate}
  \item[\normalfont(i)] $f_\theta$ satisfies
    Assumptions~\ref{ass:compact}--\ref{ass:ergodic} for all $\theta\in\Theta$, with
    a uniform lower bound $\mu_{\min}>0$.
  \item[\normalfont(ii)] $V_D \in C^\infty(\Theta)$, i.e.\ all derivatives of $V_D$
    are bounded on $\Theta$.
\end{enumerate}
\end{lemma}

\begin{proof}
The simplex is compact (Assumption~\ref{ass:compact}). Since $f_\theta$ is linear in
$\theta$, all derivatives of order $\ge 2$ vanish identically, and thus $f_\theta \in
C^\infty(\Theta)$ with trivially bounded derivatives (Assumption~\ref{ass:smooth_kernel}).
Positivity holds because $\theta_j > 0$ for some $j$ and $f_j > 0$. For ergodicity, note that $P_\theta = \sum_i \theta_i P_i \ge \theta_j P_j$ for any $j$ with
$\theta_j > 0$. Since $P_j$ is irreducible and aperiodic, for any $z, z'$ there exists
$l$ with $(P_j)^l(z'|z)>0$, so $(P_\theta)^l(z'|z) \ge \theta_j^l(P_j)^l(z'|z)>0$
and aperiodicity is inherited. Hence $P_\theta$ is ergodic. The map $\theta\mapsto d_\theta$ is continuous (since $\theta\mapsto P_\theta$ is
linear and matrix inversion is continuous on invertibles). The function
$D(s,\theta)=\sum_x d_\theta(x,s)$ is therefore continuous on the compact simplex, and
ergodicity ensures $D(s,\theta)>0$ everywhere. The Extreme Value Theorem yields
$\mu_{\min} = \min_{s,\theta}D(s,\theta) > 0$.

For part (ii), since $f_\theta$ is linear in $\theta$, the stationary averages
$\bar{r}_\theta$ and $\bar{P}_\theta$ are $C^\infty$ in $\theta$. The joint transition
$P^\theta$ and reward $r^\theta$ are compositions of $C^\infty$ maps (the softmax with
$\tau>0$ is $C^\infty$, and $Q^\star_\theta$ is $C^\infty$ by implicit differentiation
of the Bellman equation, which is well-conditioned because $\gamma < 1$). The value
function $V^\theta = (I-\gamma P^\theta)^{-1}r^\theta$ inherits $C^\infty$ smoothness
since the resolvent is an analytic function of $P^\theta$. Hence $V_D\in C^\infty(\Theta)$.
\end{proof}

\begin{theorem}[Global convergence of BO variant]
\label{thm:bo_convergence}
Under Assumptions~\ref{ass:compact}--\ref{ass:pomdp} and the convex combination
parametrization of Lemma~\ref{lem:convex_kernel}, the objective function satisfies
$V_D \in \mathcal{H}_\kappa$ with $\|V_D\|_\kappa \le B < \infty$.
Setting
$\beta_k = B + \frac{\sigma}{\sqrt{M}}\sqrt{2(\gamma_{k-1}+1+\ln(1/\varepsilon))}$
(with $\sigma^2 = \Rmax^2/(1-\gamma)^2$ from Lemma~\ref{lem:mc_variance}), the
cumulative regret of Algorithm~\ref{alg:main} with Algorithm~\ref{alg:bo} as
\textsc{StepOptimize} satisfies, with probability at least $1-\varepsilon$:
\begin{equation*}
  R_K = \sum_{k=1}^K \bigl(V_D(\theta^\star) - V_D(\theta_k)\bigr)
  \;\le\; \mathcal{O}\!\left(\sqrt{K\,\beta_K\,\gamma_K}\right),
\end{equation*}
where $\gamma_K$ is the maximum information gain of $\kappa$ on $\Theta$.
Since the Mat\'{e}rn-$5/2$ kernel has sublinear information gain ($\gamma_K = o(K)$),
average regret $R_K/K \to 0$ as $K\to\infty$, guaranteeing convergence to the global
optimum~$\theta^\star$.
\end{theorem}

\begin{proof}
The Mat\'{e}rn-$5/2$ kernel $\kappa$ on $\mathbb{R}^N$ has a reproducing kernel Hilbert
space norm-equivalent to the Sobolev space $H^{2.5+N/2}(\mathbb{R}^N)$. By
Lemma~\ref{lem:convex_kernel}(ii), $V_D \in C^\infty(\Theta)$ with all derivatives
bounded on the compact set $\Theta$, so $V_D$ admits a compactly supported $C^\infty$
extension $\widetilde{V}_D$ to $\mathbb{R}^N$. Any such extension lies in
$H^{2.5+N/2}(\mathbb{R}^N)$, and hence in $\mathcal{H}_\kappa$. Since the RKHS of
$\kappa$ restricted to $\Theta$ consists of the restrictions of functions in
$\mathcal{H}_\kappa$, with norm no larger than that of any extension, we obtain
$\|V_D\|_\kappa \le \|\widetilde{V}_D\|_{\mathcal{H}_\kappa} \triangleq B < \infty$.

To prove the regret bound, notice that by Lemma~\ref{lem:mc_variance}(iii), the evaluation noise
$\xi = \widehat{V}_D(\theta) - V_D(\theta)$ is $(\sigma/\sqrt{M})$-sub-Gaussian.
With $V_D \in \mathcal{H}_\kappa$ and sub-Gaussian noise, the standard conditions for
GP-UCB~\cite{srinivas2009gaussian} are satisfied. With the chosen $\beta_k$, the
confidence intervals
$|V_D(\theta) - \mu_{k-1}(\theta)| \le \beta_k\,\sigma_{k-1}(\theta)$ hold with
probability at least $1-\varepsilon$ uniformly over $\Theta$ and $k \ge 1$.

The UCB acquisition ensures that $\theta_k$ maximizes the upper confidence bound, meaning
$\mu_{k-1}(\theta_k) + \beta_k\,\sigma_{k-1}(\theta_k) \ge \mu_{k-1}(\theta^\star) + \beta_k\,\sigma_{k-1}(\theta^\star)$. 
Conditioned on the high-probability event that $V_D(\theta) \in [\mu_{k-1}(\theta) \pm \beta_k\sigma_{k-1}(\theta)]$ for all $\theta$, this gives an instantaneous regret of $r_k \le 2\beta_k\,\sigma_{k-1}(\theta_k)$.
Summing over $K$ iterations and applying the Cauchy--Schwarz inequality bounds the cumulative
predictive standard deviations by $\sqrt{K\gamma_K}$, yielding the cumulative regret
$R_K \le \mathcal{O}(\sqrt{K\beta_K\gamma_K})$.
For a Mat\'{e}rn-$5/2$ kernel on a domain of effective dimension $N-1$, the maximum information gain is bounded by $\gamma_K = \mathcal{O}\left(K^{\frac{(N-1)N}{5 + (N-1)N}} \log K\right)$
\cite{srinivas2009gaussian}. Because the exponent is strictly less than 1 for any finite $N$, $\gamma_K = o(K)$ and the average regret $R_K/K \to 0$.
Consequently, the simple regret
$S_K = \min_{k\le K}(V_D(\theta^\star) - V_D(\theta_k)) \le R_K/K \to 0$,
establishing exact asymptotic convergence to the global optimum $\theta^\star$.
\end{proof}

While Theorem \ref{thm:bo_convergence} guarantees exact convergence because $V_D \in C^\infty$, the repeated application of the chain rule through the Bellman matrix inversion $(I-\gamma P^\theta)^{-1}$ implies that the theoretical RKHS norm $B$ scales aggressively with the number of base kernels $N$.

\section{Numerical Experiments}
\label{sec:simulations}

We illustrate the practical utility of adaptive agent design in a non-Markovian setting, analyzing how structural choices over the internal kernels strongly dictate bounding limits on optimization and capacity.

Our continuous routing environment utilizes target observation delays, defining non-Markovian tracking properties mathematically.
The observation space is $\mathcal{O} = \{0, 1, 2, 3\}$ and action space is $\mathcal{A} = \{0, 1, 2\}$. When the agent is initialized or completes a cycle, the environment stochastically transitions to a true hidden target state $c \in \{0, 1, 2\}$ sampled uniformly. At the start of the sequence, the environment emits the true objective $o=c$. For all subsequent interim steps, prior to a routing junction, the environment exclusively outputs the null observation $o=3$. The agent must execute its internal transition rules recursively over the sequence to match its ultimate action $a$ against the hidden target $c$. The reward matrix $r(s,a,o)$ distributes biased, asymmetric payoffs to prevent trivial exploitation: choosing correctly yields $r(\cdot, a=0, \cdot) = 2.0$, $r(\cdot, a=1, \cdot) = 1.0$, and $r(\cdot, a=2, \cdot) = 0.5$ correspondingly. Incorrect routing actions distribute penalties of $-1.0$ or $-2.0$. An optimal agent must mathematically compress its long-horizon observation history dynamically across its available finite states.

We formulate candidate basis transition kernels $f_i(s'|s,a,o)$ mapping observations into the parametrized internal states recursively. Specifically, we set $f_0(s'=\min(0, |S|-1) | s, a, o=0) = 1.0$, $f_1(s'=\min(1, |S|-1) | s, a, o=1) = 1.0$, and $f_2(s'=\min(2, |S|-1) | s, a, o=2) = 1.0$. Null observations $o=3$ maintain the identity transition.
The adaptive agent transition model $f_\theta$ is defined as the strict convex combination of these basis candidate components: $f_\theta(s'|s,a,o) = \sum_{i=1}^K \theta_i f_i(s'|s,a,o)$ parametrized along the probability simplex.

\subsection{Landscape Geometry and Optimizer Performance}
We bound our empirical analysis utilizing $K=3$ distinct basis components ($f_0, f_1, f_2$) mapped across $|S|=2$ available internal states. 
By scanning the 2-degrees-of-freedom parametrized variables $(\theta_0, \theta_1)$ spanning $\theta \in \Delta^2$, we calculate the continuous total reward bounds evaluated iteratively across independent episodes (Fig.~\ref{fig:3d_landscape}). 

\begin{figure}[h]
    \centering
    \includegraphics[width=0.95\columnwidth]{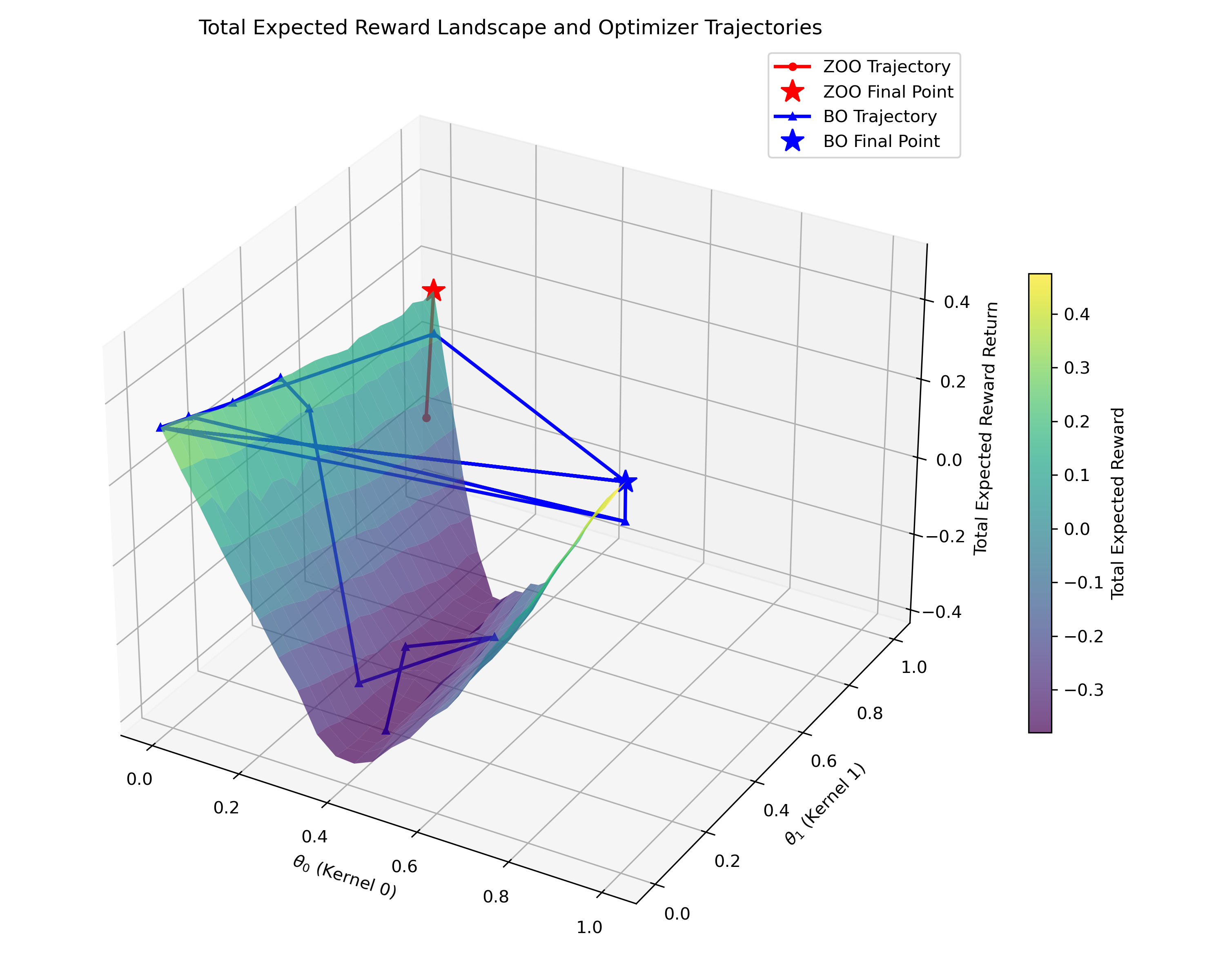}
    \caption{The evaluated total reward plot parameterized empirically over the $f_\theta$ convex combination simplex. Heavily penalized stochastic regions create strict isolation bounds.}
    \label{fig:3d_landscape}
\end{figure}

As logged in Table~\ref{tab:base_opt}, local-gradient Zeroth-Order Optimization (ZOO) actively traps within sub-optimal $f_i$ compositions. Activating global Bayesian Optimization (BO) leverages statistical exploitation to securely bypass these gradient walls, isolating the best combinatorial parameters equivalent optimally to the finite exhaustive limit calculated.

\vspace{0.2em}
\begin{table}[h]
    \centering
    \caption{Total Reward Optimizer Performance ($|S|=2$, $K=3$)}
    \begin{tabular}{lc}
        \hline
        \textbf{Method} & \textbf{Max Reward} \\ \hline
        Grid Search (Theoretical Max)  & 0.474 \\
        Bayesian Optimization (BO)     & 0.474 \\
        Zeroth-Order Ascent (ZOO)      & 0.191 \\ \hline
    \end{tabular}
    \label{tab:base_opt}
\end{table}
\vspace{0.2em}

\subsection{Impact of Kernel Parametrization}
We analyze how dynamically isolating the basis kernels $f_i$ determines the absolute global convergence maximum available structurally (Table~\ref{tab:kernels}). The base formulation of structurally divergent $f_i$ mappings provides $f_\theta$ strong capacity bounds. Reconstituting the dictionary $f_i$ completely using identically redundant copies (degenerate parameterization) uniformly eliminates all topological optimization potential, sealing maximum output equivalent strictly to identical combinations. Populating the candidate bases naively with strict independent uniform noise vectors immediately plummets the $f_\theta$ matrix directly into severely generalized negative penalties computationally independent of standard search limits, indicating the value of side information about ``good'' kernels to seed the agent design framework.

\begin{table}[h]
    \centering
    \caption{Total Reward Peak Objective across Parameter Bases}
    \begin{tabular}{lc}
        \hline
        \textbf{Kernel Component Matrices} & \textbf{Max Reward} \\ \hline
        Orthogonal Baseline Candidates  & 0.474 \\
        Identical Degenerate Mappings   & 0.462 \\
        Uniform Pure Noise Kernels      & -0.453 \\ \hline
    \end{tabular}
    \label{tab:kernels}
\end{table}

\subsection{Role of Exploratory Behavioral Policy}
Data collection fundamentally structures algorithm evaluation validity. Utilizing biased or deterministic behavioral selection mappings ($\pi_b$) drastically minimizes the state-action data trajectory breadth compiled to the discrete inner-loop planner buffer (Table~\ref{tab:policy}), algorithmically crippling the soft $Q$-learning estimates regardless of optimum candidate kernels $f_i$. Applying identically uniform random exploration naively secures standard global bounds across limited dimensionalities inherently.

\begin{table}[h]
    \centering
    \caption{Effect of Directed Trajectory Policy $\pi_b$}
    \begin{tabular}{lc}
        \hline
        \textbf{Action Sampling Bias} & \textbf{Max Reward} \\ \hline
        Uniform ($0.33, 0.33, 0.33$)      & 0.489 \\
        Minor Metric Bias ($0.6, 0.2, 0.2$) & 0.469 \\
        Severe Trajectory Bias ($0.9, 0.05, 0.05$) & 0.448 \\ \hline
    \end{tabular}
    \label{tab:policy}
\end{table}

\subsection{State Capacity Bounds}
Finally, we illustrate mathematically how limiting the generalized internal memory cardinality bounds total configuration optimization potential irrespective of available $f_i$ elements iteratively applied. By evaluating $f_\theta = \sum_{i=1}^3 \theta_i f_i$ across sequentially scaling dimensions of available internal agent state sets $|S|$, we directly plot structural memory ceilings. Formally expanding the available capacity limits maps a direct continuous scaling improvement metric optimally over $\theta$ bounds (Fig.~\ref{fig:capacity}). The total reward generated strictly climbs linearly as the matrix expands functionality, precisely matching limits and completely flattening out for bounds $|S| > 3$ exactly as logically dictated by satisfying the true native 3-environment targets fully securely. Additional dimensionality provides objectively zero generalized mathematical tracking progression.

\begin{figure}[H]
    \centering
    \includegraphics[width=0.8\columnwidth]{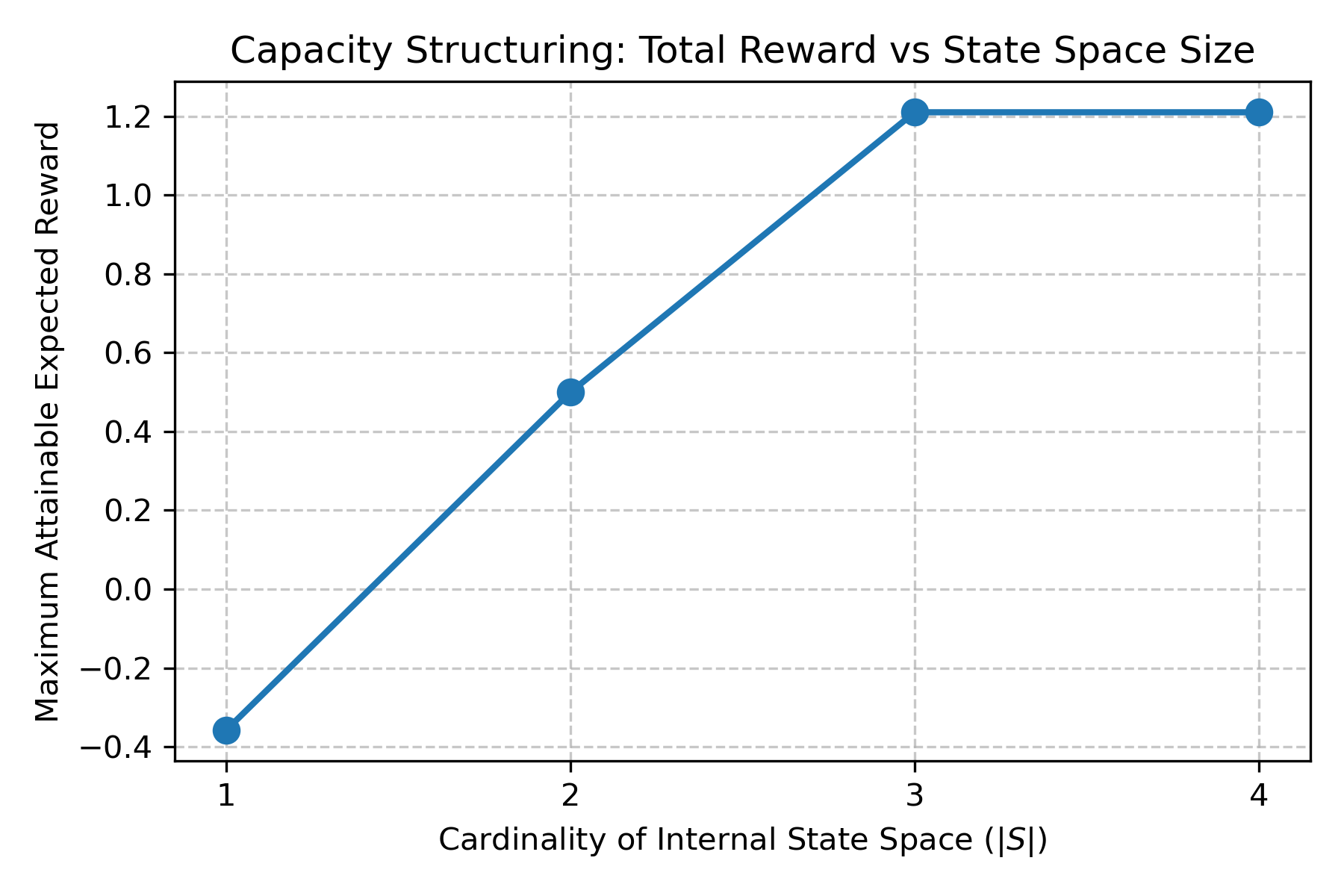}
    \caption{Total empirical expected reward evaluated monotonically maximizing iteratively as the cardinality of the internal state space linearly expands to encompass the correct environment requirements inherently.}
    \label{fig:capacity}
\end{figure}

\section{Conclusions}
\label{sec:conclusion}
In this paper, we considered the question of an agent optimizing its own transition kernel among a collection of agent states affected through its actions, playing against an unknown possibly non-Markovian environment. Given a parametrized set of transition kernels, we showed that for every choice of that parameter, a soft $Q$-iteration converges almost surely to the fixed point of a soft Bellman equation built from the stationary averages that the behavioral policy induces, and we delineated what separates the policy it yields from an optimal one. Then, we considered the question of optimizing the parameter choice through zeroth-order optimization and Bayesian optimization methods. For POMDP environments, we obtained convergence to stationary points and global sub-linear regret, respectively.

There are several interesting directions for future research. Our current study only considers the case where the dataset obtained via a behavioral policy remains fixed throughout. When the agent is learning to act against an unknown environment, it can leverage the flexibility of using newly learned policies to collect further data. We want to analyze the performance of such an adaptive agent. In our analysis, the agent keeps the cardinality of the agent state space constant. In practice, the agent can even tweak that cardinality over time. We are keen to analyze the performance of such adaptation. Finally, our work only considers a simple agent acting against an unknown but stationary non-Markovian environment. We hope to study interactions between two or more such adaptive agents playing against a known/unknown dynamic environment. This takes an alternative view point for agents acting in partial information environment as opposed to the \emph{common information} based approach \cite{nayyar2013common}, information compression \cite{mao2020information} or persuasion \cite{velicheti2025value, velicheti2023strategic, velicheti2024learning, velicheti2025harnessing, velichetilearning}.


\bibliographystyle{alpha}
\bibliography{references}

\end{document}